\documentclass{article}

\usepackage{PRIMEarxiv}

\usepackage[utf8]{inputenc} % allow utf-8 input
\usepackage[T1]{fontenc}    % use 8-bit T1 fonts
\usepackage{hyperref}       % hyperlinks
\usepackage{url}            % simple URL typesetting
\usepackage{booktabs}       % professional-quality tables
\usepackage{amsfonts}       % blackboard math symbols
\usepackage{nicefrac}       % compact symbols for 1/2, etc.
\usepackage{microtype}      % microtypography
\usepackage{lipsum}
\usepackage{fancyhdr}       % header
\usepackage{graphicx}       % graphics
\graphicspath{{media/}}     % organize your images and other figures under media/ folder

\usepackage{colortbl} % 需要在导言区加
\usepackage{venndiagram}
\usepackage{algorithm}
\usepackage{algorithmic}
\usepackage{amsmath}
\usepackage{bm}
\usepackage{booktabs}
\usepackage{multirow}
\usepackage{makecell}
\usepackage{braket}
\usepackage{amsthm}
\usepackage{xcolor}
\usepackage{booktabs}
\usepackage{threeparttable}
\usepackage{algorithm}
\usepackage{algorithmic}
\usepackage{tikz}
\usepackage{pgfplots}
\usepackage{graphicx}
\usepackage{multirow}
\usepackage{booktabs}
\usepackage{amsmath}
\usepackage{tcolorbox}
\usepackage{amsthm}
\usepackage{bm}
\newtheorem{claim}{Claim}
\usepackage{xcolor}
\usepackage{lineno,hyperref}
\modulolinenumbers[5]
\usepackage{pifont}
\title{Quantum Visual Word Sense Disambiguation: Unraveling Ambiguities Through Quantum Inference Model}

\title{Quantum Visual Word Sense Disambiguation: Unraveling Ambiguities Through Quantum Inference Model}

\author{
Wenbo Qiao\textsuperscript{1}, 
Peng Zhang\textsuperscript{2}\thanks{Corresponding author} , 
Qinghua Hu\textsuperscript{2} \\
\textsuperscript{1}School of New Media and Communication, Tianjin University, Tianjin, China \\
\textsuperscript{2}College of Intelligence and Computing, Tianjin University, Tianjin, China
}

\begin{document}
\maketitle

\begin{abstract}
Visual word sense disambiguation focuses on polysemous words, where candidate images can be easily confused. Traditional methods use classical probability to calculate the likelihood of an image matching each gloss of the target word, summing these to form a posterior probability. However, due to the challenge of semantic uncertainty, glosses from different sources inevitably carry semantic biases, which can lead to biased disambiguation results. Inspired by quantum superposition in modeling uncertainty, this paper proposes a Quantum Inference Model for Unsupervised Visual Word Sense Disambiguation (Q-VWSD). It encodes multiple glosses of the target word into a superposition state to mitigate semantic biases. Then, the 
quantum circuit is executed, and the results are observed. By formalizing our method, we find that Q-VWSD is a quantum generalization of the method based on classical probability. Building on this, we further designed a heuristic version of Q-VWSD that can run more efficiently on classical computing. The experiments demonstrate that our method outperforms state-of-the-art classical methods, particularly by effectively leveraging non-specialized glosses from large language models, which further enhances performance. Our approach showcases the potential of quantum machine learning in practical applications and provides a case for leveraging quantum modeling advantages on classical computers while quantum hardware remains immature.
\end{abstract}

\section{Introduction}

Visual word sense disambiguation (VWSD) is an emerging task that can be understood as a multimodal version of word sense disambiguation \cite{raganato2023semeval}. It aims to determine the posterior probability that an image corresponds to a target word given its context. The answer is the image that maximizes this posterior probability. A feasible approach seems to be using visual language models to represent posterior probabilities by calculating image-text matching scores. \cite{DBLP:conf/icml/RadfordKHRGASAM21,DBLP:conf/cvpr/SinghHGCGRK22}. However, VWSD focuses on more complex task scenarios where the target words are ambiguous and polysemous, and the candidate images can easily be confused. Simply relying on matching methods cannot effectively resolve this task.
\cite{DBLP:conf/blackboxnlp/RassinRG22,DBLP:conf/acl/KwonGL0023}.

Some gloss-enhanced methods for word sense disambiguation have inspired related work to utilize external knowledge for VWSD \cite{huang2019glossbert,bevilacqua2021recent,zhang2024quantum}.
The disambiguation process can be viewed as a weighted summation based on classical probability where answer images are not directly determined by the context but by hidden variables related to glosses. The conditional probability of the image is calculated independently for each gloss. The final posterior probability is then derived using the law of total probability. 
However, due to the inherent semantic uncertainty in natural language, glosses from different perspectives (or dictionaries) inevitably carry semantic biases \cite{misono1997effects}. 
The \textbf{A} in Fig. \ref{fig:enter-label} quantifies the semantic bias \footnote{Fidelity is defined as the square of the inner product of two normalized vectors. The closer the fidelity is to 1, the more similar the two vectors are. It is commonly used in quantum information and can be considered a special form of cosine similarity.} .
Directly using these specific glosses for semantic disambiguation can lead to biased results \cite{zhang2024quantum}. For example, as shown in Fig. \ref{fig2}, using glosses from WordNet and large language models (LLMs) for ``\textit{Andromeda tree}" results in different disambiguation outcomes. This is because WordNet glosses focus more on the botanical meaning of ``\textit{Andromeda}", while LLMs glosses emphasize its astronomical meaning.

\begin{figure}[t]
    \centering
    \includegraphics[width=0.5\linewidth]{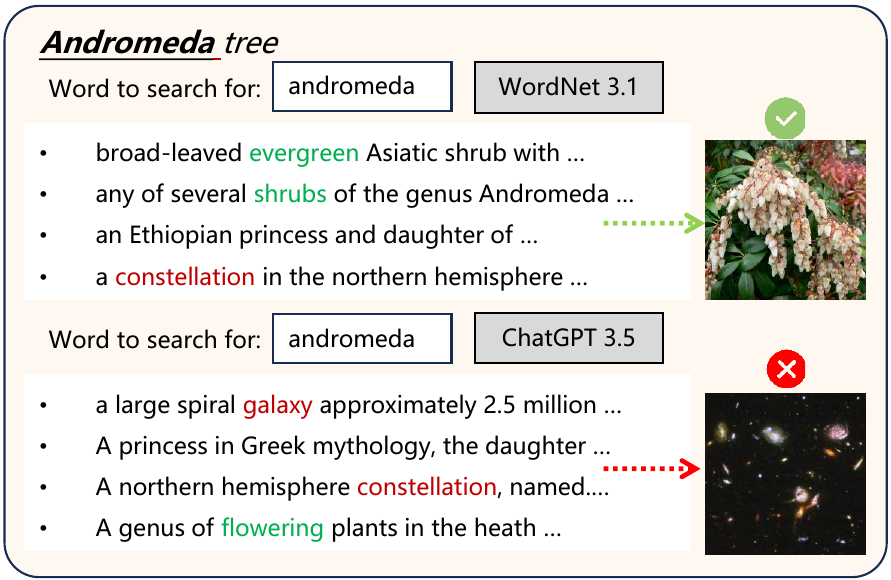}
    \caption{Different sources of glosses focus on various perspectives, which can lead to erroneous VWSD results.}
    \label{fig2}
\end{figure}

We suggest the primary reason underlying the aforementioned problems is that these models attempt to use glosses of certainty to represent the natural language of uncertainty \cite{goodman2015probabilistic}. Inspired by the principle of quantum superposition \cite{von2018mathematical}, we attempt to model the uncertainty of target words not with specific glosses but with a superposition of their glosses. When candidate images serve as observables, the superposition state of the target word dynamically collapses to implicit semantics based on observations rather than static semantics derived from specific glosses, thereby reducing semantic bias. The \textbf{B} in Fig. \ref{fig:enter-label} quantifies the effect of using quantum superposition states to mitigate semantic bias. To model this process, it is natural to use quantum machine learning to model and measure the superposition state of the target word and achieve visual word sense disambiguation.

Quantum machine learning (QML) has recently gained attention for its unique advantages, particularly in fitting capabilities \cite{schuld2021effect,yu2022power,shin2023exponential} and modeling generative models \cite{rudolph2022generation,qiao2024quantum2}. However, how these advantages can be applied to more practical tasks remains to be explored \cite{schuld2021effect}. As the most natural tool for describing and measuring superposition states, we are consequently motivated to use quantum machine learning to resolve ambiguities in visual vocabulary, while also exploring its potential advantages in more practical downstream tasks. 

Thus, we propose a Quantum Inference Model for Unsupervised Visual Word Sense Disambiguation (Q-VWSD).
Specifically, we use a pre-trained model to encode the context, images, and glosses. We then construct a superposition state representation of the glosses and encode it into the quantum circuit. Given an image, an observation is established, and the quantum circuit is executed. The output is subsequently interpreted as the posterior probability.
The maximum posterior probability is selected as the answer. We formalize Q-VWSD further and find that it violates the law of total probability, serving as a quantum generalization of classical probability methods. Based on this, we also propose a heuristic version of Q-VWSD, which is theoretically equivalent to the quantum circuit but better suited for running on classical computers. We conducted experiments on unsupervised visual word sense disambiguation, and the results show that our method outperforms the latest classical methods. 
\textbf{In summary, the main contributions of this paper are as follows:
}
% Additionally, We found that Q-VWSD can effectively utilize non-specialized glosses generated by large language models. In contrast, classical methods are unable to utilize these non-specialized glosses.

\begin{itemize}
\item We designed Q-VWSD, which reduces semantic bias through superposition state representation and measurement, surpassing the performance of classical models.

\item We formalized Q-VWSD and analyzed the sources of its advantages, highlighting the fundamental differences between Q-VWSD and classical probabilistic methods.

\item We demonstrated the modeling advantages of QML in practical scenarios, which further led us to propose a heuristic version more suitable for classical computing, showing that QML can guide the design of better artificial intelligence even without quantum hardware.
\end{itemize}
In the following content, Section \ref{c2} will introduce the research area of this work, followed by the background knowledge in Section \ref{c3}. Our model will be presented in Section \ref{c4}, and experimental validation will be provided in Section \ref{c5}.

\begin{table*}[]
\centering
\label{t1}
\footnotesize
\tabcolsep=0.1cm

\caption{Quantum machine learning related work.}
\begin{threeparttable}
\begin{tabular}{@{}cccccc@{}}
\toprule
\textbf{Model}                              & \textbf{Executable}& \textbf{Comparable} & \textbf{Cross-Modal} & \textbf{Task}   & \textbf{Foundation} \\ \midrule
\multirow{3}{*}{Quantum Inspired Methods}  &          \ding{53}           &      \ding{51}                &     \ding{53}                  & Matching         & \multirow{3}{6cm}{\raggedright Using quantum probability to solve classical problems by analogy with quantum mechanical phenomena}.    \\
                                            &            \ding{53}         &    \ding{51}                  &             \ding{51}       & Classification    &       \\
                                            &          \ding{53}           &   \ding{51}                   &         \ding{53}   & Retrieval      &          \\ \midrule
\multirow{3}{*}{Quantum Computing Methods} &     \ding{51}                &   \ding{53}                  &          \ding{53}              & Regression      & \multirow{3}{5cm}{\raggedright Encoding classical data into quantum circuits and solving classical problems via quantum circuits.

}\\
                                                                 &     \ding{51}                &   \ding{53}                  &          \ding{53}         & Classification    &       \\
                                                              &     \ding{51}                &   \ding{53}                  &          \ding{53}                                                                                                                 & Generation      &        \\ 
                                    \hline  
                                                              
                                                          \multirow{3}{*}{Q-VWSD (ours)}
                                                              &     \multirow{3}{*}{\ding{51}}                &   \multirow{3}{*}{\ding{51}}                  &          \multirow{3}{*}{\ding{51}}                                                                                                                 & \multirow{3}{*}{Inference}      &  \multirow{3}{6cm}{\raggedright Encoding classical data with quantum probability and solving classical problems via quantum circuits.}      \\[3.5ex] 
                                                              \bottomrule

\end{tabular}
\textbf{Executable} indicates whether it can run on a quantum computer; \textbf{Comparable} denotes whether it can achieve performance comparable to classical SOTA models at the time; \textbf{Cross-Modal} specifies whether it supports cross-modal functionality; \textbf{Task} refers to the applications or tasks it addresses; and \textbf{Foundation} indicates the foundational principles utilized. %添加此处
\end{threeparttable} %添加此处
\end{table*}

\begin{figure}
    \centering
    \includegraphics[width=0.5\linewidth]{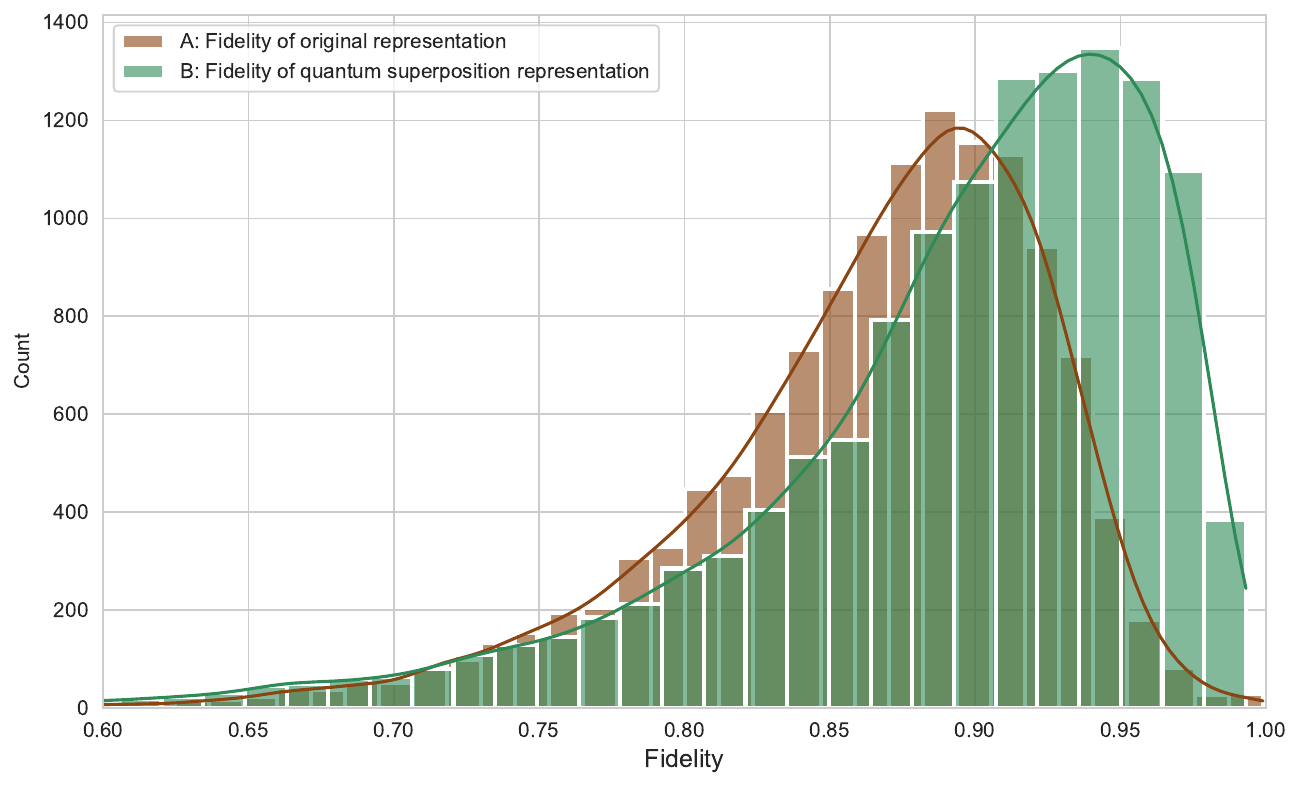}
    \caption{Fidelity between the target words constructed using glosses from WordNet and those constructed using glosses from ChatGPT 3.5. \textbf{A} shows that a large number of samples did not achieve 90\% fidelity, indicating a semantic bias between WordNet and  ChatGPT 3.5. \textbf{B} shows that after representation with quantum superposition states, the fidelity between WordNet and ChatGPT 3.5 reached 95\%, suggesting hope for mitigating semantic bias.}
    \label{fig:enter-label}
\end{figure}

\section{Related Work}\label{c2}
This subsection aims to establish a connection between VWSD and QML. Although this link is not commonly explored in related work, we believe that QML can elegantly address VWSD's challenges.

\subsection{Quantum Machine Learning}

Leveraging the unique properties of quantum computing to develop QML that surpasses classical machine learning has been a key focus of attention. Related work has explored various directions, such as quantum time series prediction \cite{DBLP:conf/nips/Bausch20,chen2022quantum}, quantum graph neural networks \cite{DBLP:conf/kdd/YanTY22,DBLP:conf/nips/TangY22}, and quantum generative models \cite{rudolph2022generation,qiao2024quantum2}. Recent studies have investigated the strong fitting capability of QML \cite{zhao2024quantum}, starting from the universal approximation theorem for quantum models \cite{schuld2021effect}.
However, due to the immaturity of quantum computing hardware, the advantages of quantum methods are often difficult to apply in machine learning, and their further implementation and practicality are not seen as promising. Meanwhile, another broad category of quantum-inspired machine-learning methods has also developed. These methods have been successfully applied to tasks such as information retrieval \cite{sordoni2013modeling, jiang2020quantum} and question answering \cite{zhang2018end, zhang2022complex}, as well as modeling tasks related to cognitive phenomena such as human emotions and metaphors \cite{gkoumas2021quantum,qiao2024quantum}. However, most of these models merely draw analogies between quantum phenomena and specific problems. 
Especially in pursuit of effectiveness, these methods tend to overlook certain principles of quantum mechanics, leading to their rigor being often questioned. 
We summarize the main basic information of the mentioned work in Table \ref{t1}. Our method can be seen as a combination of Quantum Inspired and Quantum Computing methods.

\subsection{Visual Word Sense Disambiguation}
VWSD is an advanced version of word sense disambiguation, recently introduced in the SemEval 2023 challenge \cite{raganato2023semeval}. Popular image-text retrieval and visual language models have not performed well on this task \cite{DBLP:conf/cvpr/SinghHGCGRK22,DBLP:conf/icml/RadfordKHRGASAM21}. Consequently, numerous studies have explored supervised methods \cite{wei2023stfx,li2023ecnu_miv}, which enhance model performance by incorporating external knowledge through techniques like prompt learning \cite{zhang2023srcb,ghahroodi2023sut}, back-translation \cite{vaiani2023polito}, and fine-grained contrastive learning \cite{yang2023tam}. However, these models still rely on expensive labeled data and lack support for unsupervised VWSD. Recently, unsupervised VWSD methods based on large language models have been proposed \cite{kritharoula2023large,dadas2023opi}, directly using large language models to read context and images for disambiguation. These models' inference processes are opaque and may suffer from hallucinations. \cite{DBLP:conf/acl/KwonGL0023} suggested using Bayesian inference combined with glosses to achieve unsupervised VWSD. However, this approach is vulnerable to semantic biases in the glosses.

Building on the mentioned work, our study aims to leverage the inherent superposition state modeling capability of QML to represent glosses, mitigate semantic bias, and address challenges in VWSD. Compared to existing QML, we further explore the advantages of modeling, tackling more realistic tasks, particularly demonstrating how QML inspires us to design better classical machine learning models.

\section{Preliminaries}\label{c3}

This section will briefly introduce the basic knowledge of quantum computing and quantum probability. Mathematically, they primarily utilize linear algebra to describe the transformation rules of complex-valued vectors within a Hilbert space (i.e., a complex-valued high-dimensional space $\mathcal{H}\in \mathcal{C}^d$). 

\subsection{Quantum Computing}

\textbf{Quantum State.} Unlike a classical bit, which is non-zero or one, a quantum bit (qubit) can be in a superposition state of 0 and 1, which has a mathematical form:
\begin{equation*}
|\phi\rangle = \alpha|0\rangle + \beta|1\rangle,
\end{equation*}
where $|0\rangle=\begin{bmatrix}1,0\end{bmatrix}^T$ and $|1\rangle=\begin{bmatrix}0,1\end{bmatrix}^T$are a set of orthonormal basis, corresponding to the basis states. $\alpha$ and $\beta$ are complex numbers that satisfy $|\alpha|^2+|\beta|^2=1$. So they are called probability amplitudes.
Mathematically, a qubit is actually represented as a 2-dimensional vector.
Similarly, a $N$-qubit quantum circuit can be represented by a quantum state in a $2^N$-dimensional Hilbert space: 
\begin{equation*}
    |\phi\rangle = \sum_{i=1}^{2^N}\alpha_i|x_i\rangle, 
\end{equation*}
where $\{a_i\}_{i=1}^{2^N}$ denotes the set of probability amplitude satisfying $\sum_i |a_i|^2=1$. $\{|x_i\rangle\}_{i=1}^{2^N}$ is a set of orthonormal basis in $2^N$-dimensional Hilbert space, corresponding to the basis states, where $x_i\in \{0,1\}^n$.

\begin{table}[t]
\centering
\caption{Notations and Descriptions}
\begin{tabular}{@{}cc@{}}
\toprule
\textbf{Notation} & \textbf{Descriptions}                                              \\ \midrule
$|{\phi}\rangle$ & Column vector, also called ket                            \\
$\langle{\phi}|$ & Conjugate transpose of $|\phi\rangle$, also called bra  \\
$\alpha$ & Probability amplitude of quantum state                           \\ 
P & Projection operator                            \\ 
$\dagger$ & Conjugate transpose \\
$M$ & Measurement operator \\
$\bm{v}$ & The encoded representation of image  \\
$\bm{c}$ & The encoded representation of context \\
$t$ & The target word                            \\ 
$w$ & The encoded representation of word\\
$\bm{g}$ & The encoded representation of the target word's gloss                      \\
$O$ & Observation                            \\ 
$|G\rangle$ & The superposition state representation of the target word's glosses                            \\ 
$cos (\theta)$ & The phase in the interference process                            \\ 
\bottomrule
\end{tabular}
\label{t2}
\end{table}

\textbf{Quantum Evolution.} 
Like classical circuits, quantum circuits can manipulate the state of qubits through quantum gates. Essentially, they can be regarded as unitary evolution from $t$ moment to $t'$:
\begin{equation*}
|\phi'\rangle = U|\phi\rangle.
\end{equation*}
In quantum computing, one focuses on quantum gates, which are also a type of unitary matrix. For example, the most commonly used Pauli gates:
\begin{equation*}
    \begin{split}
    \mathrm{X}=\left(\begin{array}{ll}
0 & 1 \\
1 & 0
\end{array}\right), \mathrm{Y}=\left(\begin{array}{cc}
0 & -i \\
i & 0
\end{array}\right), \mathrm{Z}=\left(\begin{array}{cc}
1 & 0 \\
0 & -1
\end{array}\right).
    \end{split}
\end{equation*}

\textbf{Quantum Measurement.} Finally, given a set of measurement operators \( \{M_{x_i}\} \), the measurement result can be read at the end of the quantum circuit. The probability of measuring the result \( x_i \) is:
\begin{equation}
\label{e-1}
P(x_i) = \langle \phi' | M_{x_i}^\dagger M_{x_i} | \phi' \rangle.
\end{equation}
After measurement, the quantum system collapses to:
$
|\phi'_{x_i}\rangle = \frac{M_{x_i} |\phi'\rangle}{\sqrt{P(x_i)}}.
$
This process is complete because $I=\sum_i M_{x_i}^\dagger M_{x_i}$.
For example, when using the measurement operators \( \{M_0, M_1\} \), the probabilities of measuring 0 and 1 will be:
$
P(0) = \langle \phi' | M_0^\dagger M_0 | \phi' \rangle, \quad P(1) = \langle \phi' | M_1^\dagger M_1 | \phi' \rangle
$.
Clearly, the measurement results are randomly determined by the above probabilities. Therefore, we are more concerned with the deterministic statistical results. For instance, given an observation, the expected value of the quantum circuit can be obtained:
\[
\langle O \rangle = \langle \phi' | O | \phi' \rangle,
\]
where $O$ is an observation that satisfies the Hermitian matrix property (e.g., Pauli operators $\mathrm{Z}$). $\langle \phi' |$ denotes the conjugate transpose of $|\phi' \rangle$ (i.e., $\langle \phi' | = |\phi' \rangle^\dagger$).

\subsection{Quantum Probability Theory}

In classical probability, an event is a subset of the sample space. However, in quantum probability, an event is a subspace in the Hilbert space. As shown in Eq. (\ref{e-1}) above, quantum measurements can calculate the probability of the occurrence of an event. We primarily focus on a subclass of quantum measurements known as projection measurements, which can describe conditional probabilities. Suppose there are two states $|e_1\rangle$ and $|e_2\rangle$ representing independent events $e_1$ and $e_1$, respectively. The conditional probability of the event $e_1$ given the event $e_2$ is:
\begin{equation*}
    P(e_1|e_2)=\langle e_2 |\mathrm{P}_{e_1}| e_2 \rangle = ||\mathrm{P}_{e_1}| e_2 \rangle||^2,
\end{equation*}
where $\mathrm{P}_{e_1}=| e_1 \rangle\langle e_1 |$ denotes a measurement operator. Given any two quantum states $|e_1\rangle$ and $|e_2\rangle$, the inner product and outer product can be defined as $\langle e_1|e_2\rangle$ and $|e_1\rangle \langle e_2|$, respectively.  We have described the notations used in this paper in Table \ref{t2}. It is important to note that the implementation of quantum probability does not require specific quantum hardware. These properties of quantum probability can be manually implemented on classical computers. However, on quantum circuits, these properties will naturally be realized without manual implementation.

\section{Quantum Visual Word Sense Disambiguation}\label{c4}
In this section, we will introduce Q-VWSD and derive its formal representation, demonstrating that it is a quantum extension of classical probability methods. Based on this, we also propose an equivalent quantum-inspired version of Q-VWSD to better analyze the model and move towards practical applications.

\subsection{Task Definition of Unsupervised VWSD}
VWSD requires selecting the 
optimal image $v\in V$ from a set of candidates $V$. The chosen image $v$ should most accurately reflect the meaning of the target word $t$ within its given context $c$. Therefore, the task requires the consideration of a conditional probability:
\begin{equation}\label{e1}
    v = \arg\max_{v \in V}  P(v|c),
\end{equation}
where $t \in c$.
In unsupervised learning scenarios, a conventional approach is to separately encode image and text features. The posterior probability is then computed based on the similarity scores between the image and text features.

\subsection{Quantum Inference Circuit}

Q-VWSD can be implemented through a quantum inference circuit. The overall structure of the quantum inference circuit is shown in Fig. \ref{fig3}. It can be divided into two parts: constructing superposition state representations and measuring these superposition states.

\textbf{Representation.} Given a set of glosses $G = \{g_1, g_2, \dots, g_n\}$, we use a pre-trained text encoder to encode each gloss and utilize the encoding corresponding to the $[CLS]$ as the output for each gloss:
\begin{equation}
    \begin{split}
            \bm{g}_1,\bm{g}_2,\dots,\bm{g}_n = &TextEncoder(g_1,g_2,\dots,g_n),
            \end{split}
\end{equation}
where $g_i=\{[CLS],w_{i1},w_{i2}, \dots , w_{il}\}$. The gloss is then normalized to ensure that the representation satisfies the properties of quantum states:
\begin{equation}\label{e3}
           |g_i\rangle = \frac{\bm{g}_i}{\sqrt{{\bm{g}_i}^2}}.
\end{equation}

According to the principle of state superposition, the multiple glosses are represented as the superposition of each gloss based on the probability amplitude:
\begin{equation}\label{e4}
           |G\rangle = a_1|g_1\rangle + a_2|g_2\rangle + \dots + a_n|g_n\rangle,
\end{equation}
where $[a_1,a_2,\dots,a_n]$ denotes the set of probability amplitude satisfying  $\sum_i |a_i|^2=1$. 
After preparing the superposition state 
$|G\rangle$, we can encode it into the quantum circuit using amplitude encoding techniques \cite{nakaji2022approximate}. 

\begin{figure}[t]
    \centering
    \includegraphics[width=0.7\linewidth]{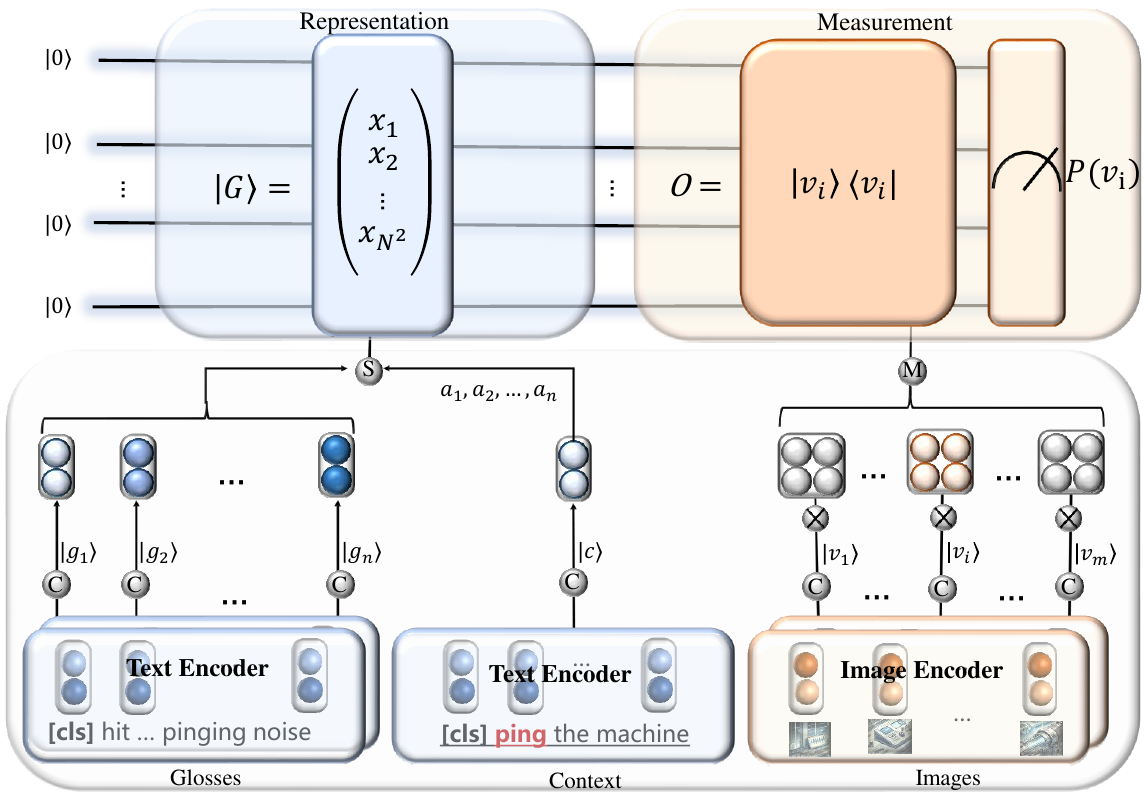}
    \caption{Main structure of quantum inference circuit (Q-VWSD$_{QC}$). \textbf{C} denotes Eq. (\ref{e3})/(\ref{e6}) and indicates the initialization of a quantum state. \textbf{S} denotes Eq. (\ref{e4}) and indicates the construct of the superposition states of the glosses. \textbf{M} denotes Eq.(\ref{e10}) and indicates the construct of the observation.}
   \label{fig3}
\end{figure}

The probability amplitude serves a role similar to a weighting factor, acting as a superposition coefficient. Intuitively, the more relevant a gloss is to the context, the larger the coefficient should be assigned. Therefore, we need to encode the relationship between context and gloss further. Similarly, given a context $c=[[CLS],w_1,w_2,\dots,w_l]$ and $t\in c$, we can encode the context information via a pre-trained text encoder, and the encoding corresponding to the target word is the output:
\begin{equation}
    \bm{c} = TextEncoder([[CLS],w_1,w_2,\dots,w_l]|t).
\end{equation}
Then, we normalize the encoding so that it satisfies the properties of quantum states:
\begin{equation}\label{e6}
    |c\rangle = \frac{\bm{c}}{\sqrt{\bm{c}^2}}.
\end{equation}
The correlation between the $i$-th gloss and context can be calculated by normalizing the inner product:
\begin{equation}\label{e7}
    P(g_i|c)=a_i^2 =\frac{|\langle g_i | c\rangle|^2}{\sum_i |\langle g_i | c\rangle|^2}.
\end{equation}
This process acts like an attentional mechanism that provides a weighted prior about the context for the superposition of glosses.

\textbf{Measurement.} Given a set of candidate images $V=\{v_1, v_2,\dots,v_m\}$, we obtain a representation of the image using a pre-trained image encoder:
\begin{equation}
\bm{v}_1,\bm{v}_2,\dots,\bm{v}_m = ImageEncoder(v_1,v_2,\dots,v_m).
\end{equation}
Each image will be normalized and then undergo an outer product operation to form the corresponding observation:
\begin{equation}\label{e9}
           |v_i\rangle = \frac{\bm{v}_i}{\sqrt{{\bm{v}_i}^2}},  \quad 
    O_{v_i}=|v_i\rangle \langle v_i|.
\end{equation}
By executing the quantum circuit, we can obtain the expectation value under a given observation:
\begin{equation}\label{e10}
           \langle O_{v_i} \rangle =\langle G|O_{v_i}|G\rangle.
\end{equation}
Since this observation is also a projection operator, the expectation value output by the quantum circuit represents the posterior probability of the image in relation to context:
\begin{equation}\label{e10}
           v = \arg\max_{v \in V}  P'(v|c) = ||\mathrm{P}_{v}|G\rangle||^2=\langle G|O_v|G\rangle .
\end{equation}
The image corresponding to the observation that maximizes the posterior probability is the result for VWSD \footnote{$||\mathrm{P}_{v}|G\rangle||^2=\langle G|v\rangle\langle v|v\rangle \langle v|G\rangle=\langle G|O_v|G\rangle$}.

\begin{figure*}[t]
    \centering
    \includegraphics[width=1\linewidth]{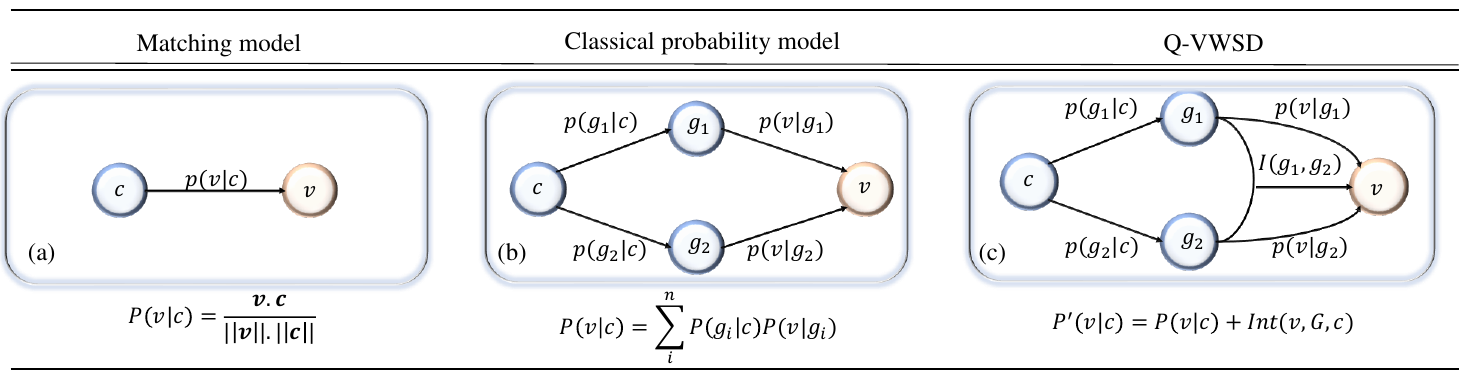}
    \caption{Q-VWSD vs. classical methods. (a) Traditional image-text matching directly computes the cosine similarity. (b) The classical probabilistic model uses  the law of total probability and applies Bayesian inference to accumulate evidence, thereby obtaining the posterior probability. (c) Our Q-VWSD employs quantum probability, and the ``interference effects" between different glosses must also be considered.}
    \label{ff3}
\end{figure*}

\subsection{Formalization of Quantum Inference Circuit}

To explain the advantages of Q-VWSD, we formally analyze the quantum inference model using quantum probability.
The matching method is illustrated in Fig. \ref{ff3} (a). A typical example is common image-text matching methods. However, these methods cannot be directly applied to VWSD because they overlook the latent variables associated with the ambiguity of the target words. From the perspective of classical probability, VWSD can be understood as a weighted summation process, as shown in Fig. \ref{ff3} (b). For instance, \cite{DBLP:conf/acl/KwonGL0023} calculates Eq. (\ref{e1}) using the law of total probability as follows:
\begin{equation}
    \begin{split}\label{e11}
            P(v|c) &= \sum_i^n P(g_i|c)P(v|g_i).
    \end{split}
\end{equation}

From the perspective of quantum probability, the posterior probability of an answer image can be viewed as a projective measurement process. Q-VWSD thus has the following properties:
\begin{claim} \label{claim1}
Q-VWSD violates the law of total probability and is a quantum generalization of the VWSD methods based on classical probability. 
\end{claim}
\begin{proof}
Recalling Eq. (\ref{e10}) again, we obtain:
\begin{equation}
    \begin{split}\label{e12}
            P^{'}(v|c)  &= ||\mathrm{P}_v|G\rangle ||^2\\
&=||\mathrm{P}_v(a_1|g_1\rangle +a_2|g_2\rangle +\dots+a_n|g_n\rangle )||^2\\
&=\sum\nolimits_i^n(\langle g_i|a_i^*)\mathrm{P}_v\sum\nolimits_j^n(a_j|g_j\rangle)\\
&=   \sum\nolimits_i^n \sum\nolimits_j^n a_i^* a_j\langle g_i|\mathrm{P}_v|g_i\rangle\\
&=  \sum\nolimits_i^n a_i^2 \langle g_i|\mathrm{P}_v|g_i\rangle+\sum\nolimits_{j\neq i}^n a_i^* a_j\langle g_i|\mathrm{P}_v|g_j\rangle\ \\
&=  \sum\nolimits_i^n a_i^2 \langle g_i|\mathrm{P}_v|g_i\rangle\\
&+\sum\nolimits_{j\neq i}^n 2a_i a_j|\langle g_i|\mathrm{P}_v|g_j\rangle|cos(\theta_{ij}).
    \end{split}
\end{equation}
Here $cos (\theta_{ij})$ is the phase of $\langle g_i|\mathrm{P}_v|g_j\rangle$. $\langle g_i|\mathrm{P}_v|g_i\rangle$ denotes the conditional probability $P(v|g_i)$ that the gloss $g_i$ projects onto the image $v$. Since $a_i$ is real, $a_i$ equals its conjugate complex number $a_i^*$. Therefore, bringing $P(v|g_i)=\langle g_i|\mathrm{P}_v|g_i\rangle$ and $P(g_i|c)=a_i^2$ into Eq. (\ref{e12}) has:
\begin{equation}
\begin{split}\label{e13}
 P^{'}(v|c) &= \sum\nolimits_i^n P(g_i|c)P(v|g_i)+Int(v,G,c)\\
 & = P(v|c) + Int(v,G,c).
 \end{split}
\end{equation}
By inspecting Eq. (\ref{e13}), it is found that the posterior probability produces an additional interference term $Int(v,G,c)=\sum\nolimits_{j\neq i}^n 2a_i a_j|\langle g_i|\mathrm{P}_v|g_j\rangle|cos(\theta_{ij})$, which indicates that Q-VWSD violates the law of total probability. By comparing with Eq. (\ref{e11}), we observe that Q-VWSD reduces to a classical probability method when the interference term equals zero. Thus, Q-VWSD can be regarded as a quantum generalization of the classical method. Fig. \ref{ff3} (c) illustrates the differences between Q-VWSD and classical methods.
\end{proof}

\textbf{\textit{Analysis}:} The interference term is the main factor that differentiates Q-VWSD from the classical methods. The classical method directly produces explicit static semantics from specific glosses, which are easily affected by semantic bias in the glosses.
But, Q-VWSD encodes glosses as superposition states in a quantum circuit, creating interference effects during measurement. The effect captures semantic interactions within glosses, dynamically adjusting semantics based on observations and thereby reducing semantic bias \cite{zhang2024quantum}. In quantum cognition \cite{busemeyer2012quantum},  interference effects can also be used to explain various human cognitive phenomena. For example, humans often do not linearly consider each possible meaning of the target words and do not accumulate evidence for image-text matching. Instead, after reading the text, they simultaneously comprehend multiple meanings of the target word and make a rapid judgment to draw a conclusion based on intuition. Therefore, in VWSD, the quantum probability model may better align with human cognitive phenomena than the classical probability model.

\subsection{Quantum-Inspired Inference Model}

\begin{figure*}[t]
    \centering
    \includegraphics[width=0.9\linewidth]{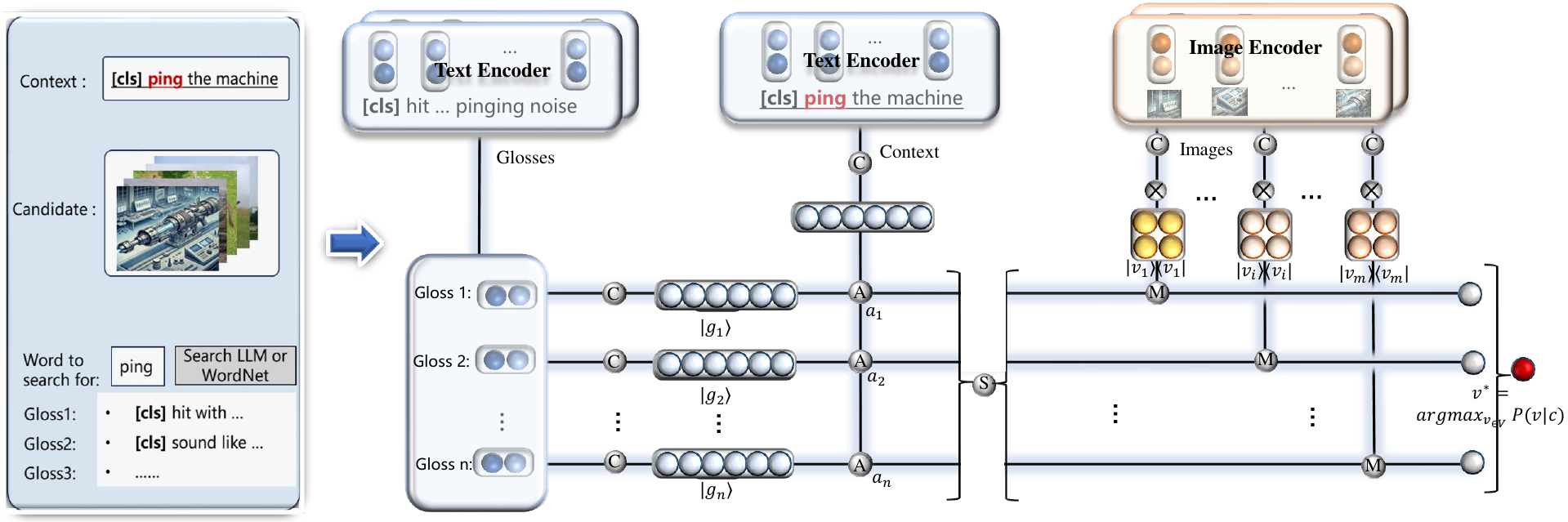}
    \caption{Main structure of quantum-inspired inference model (Q-VWSD$_{QI}$). \textbf{C} denotes Eq. (\ref{e3})/(\ref{e6}) and indicates the initialization of a quantum state. \textbf{A} denotes Eq. (\ref{e7}) and refers to calculating the probability amplitude. \textbf{S} denotes Eq. (\ref{e4}) and indicates the construct of the superposition states of the glosses. \textbf{M} denotes Eq.(\ref{e10}) and indicates quantum measurement.}
    \label{f1}
\end{figure*}

As seen from Eq. (\ref{e13}), the advantages of the quantum inference circuit lie not in computation but in modeling. This naturally leads us to design a quantum-inspired model based on the principles of Q-VWSD. This model can explicitly model interference terms, allowing for more flexible execution on classical computers. To distinguish between the models, we denote our model based on the quantum inference circuit as Q-VWSD$_{QC}$ and the quantum-inspired inference model as Q-VWSD$_{QI}$. Since these two models are theoretically equivalent, we can use the latter in ablation experiments to analyze the properties when analyzing the former becomes challenging. 

As shown in Fig. \ref{f1}, we can calculate the mathematical form of Q-VWSD through manual analytical solutions. The process can be directly simplified as follows:
\begin{equation}
    \begin{split}\label{e14}
            P^{'}(v|c)  
&=  \sum\nolimits_i^n a_i^2 \langle g_i|\mathrm{P}_v|g_i\rangle\\
&+\sum\nolimits_{j\neq i}^n 2a_i a_j|\langle g_i|\mathrm{P}_v|g_j\rangle|cos(\theta_{ij}).
    \end{split}
\end{equation}
Yet, to implement Q-VWSD$_{QI}$ according to Eq. (\ref{e14}), it is necessary to model the phase $\cos(\theta_{ij})$ within the interference terms explicitly.
The phase influences the interference effect and can convey important information such as emotions or sarcasm \cite{gkoumas2021quantum}, or even positional information \cite{zhang2022complex}. When the model only considers the real-valued vector of $|g_i\rangle$, $\cos(\theta_{ij}) = 1$. When the model considers the complex-valued vector of $|g_i\rangle$, $\cos(\theta_{ij}) \in [-1, 1]$. Since the complex-valued features only influence the final phase, we can manually assign specific values to the phases without needing to encode the semantics into the high-dimensional complex-valued Hilbert space. For simplicity, following some works \cite{busemeyer2012quantum}, we directly use the cosine similarity between glosses to represent $\cos(\theta_{ij})$. Under this setting, when $|g_j\rangle$ and $|g_i\rangle$ are orthogonal, the interference term equals zero, causing Q-VWSD to degenerate into the classical probability method. This indicates that in classical probability, semantic bias in glosses can only be avoided by employing completely orthogonal glosses with no semantic overlap. However, in practical scenarios, glosses of target words are rarely completely independent. This highlights the necessity of designing Q-VWSD.

\subsection{Glosses from Large Language Models}
In addition to the semantic bias of the glosses, their completeness will also affect the disambiguation effects. For example, when an out-of-vocabulary phenomenon occurs, the glosses become incomplete.
Theoretically, large language models have wider knowledge coverage and more complete semantic information, so we further utilize the annotation capability of large language models to generate glosses. 
As mentioned earlier, since the generated glosses are not guided by human experts, they may contain more significant semantic biases. Therefore, directly using glosses from large language models in classical probability methods is challenging.
In contrast, Q-VWSD can alleviate the impact of semantic bias and leverage the more complete semantic features within large language models.
Note that this differs from methods \cite{kritharoula2023large,dadas2023opi} using large language models for VWSD, where context and images are provided. In our approach, large language models are used solely as a dictionary without context or image inputs. This reasoning process is both interpretable and efficient. The prompt templates we use are as follows:

\begin{tcolorbox}[colback=gray!10!white, colframe=black, title=Prompt]
\textbf{Content:}\\

You are an excellent dictionary. Please provide at least five concise and precise interpretations of a given word. Provide glosses that do not overlap as much as possible.\\

\textbf{Input format:} \\

[$word_1$],[$word_2$],...,[$word_n$]\\

\textbf{Output format:} \\

        {$word_1$}: [[$gloss_1$], [$gloss_2$],..., [$gloss_5$]], \\
        {$word_2$}: [[$gloss_1$], [$gloss_2$],..., [$gloss_5$]], \\
                       ..., \\
        {$word_n$}: [[$gloss_1$], [$gloss_2$], ..., [$gloss_5$]].\\

\end{tcolorbox}

\begin{figure*}[t]
    \centering
    \includegraphics[width=0.9\linewidth]{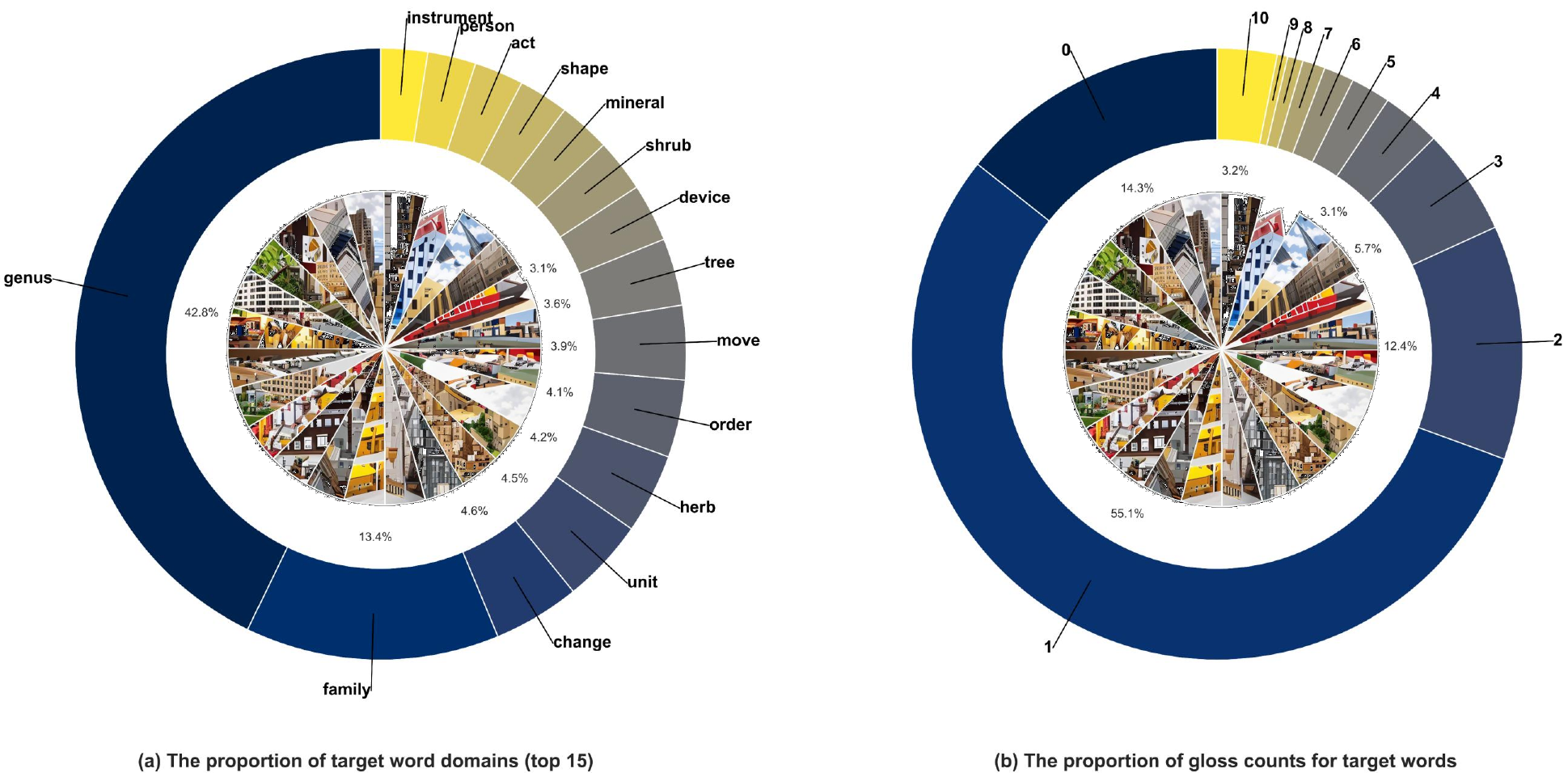}
    \caption{Presentation of the SE23 dataset overview. (a) The proportion of target words in the domain of WordNet. (b) The proportion of the number of glosses of target words by WordNet.
}
    \label{f5}
\end{figure*}

\section{Experiments}\label{c5}

In this section, we will conduct experimental validation, focusing on addressing the following research questions:

 \begin{itemize}
     \item \textbf{RQ 1:} Can Q-VWSD surpass SOTA models based on classical probability, and can it alleviate semantic bias?
     \item\textbf{RQ 2:} What is Q-VWSD's performance in the face of polysemous words?
     \item\textbf{RQ 3:} Which components in Q-VWSD are most crucial to its performance?
    \item\textbf{RQ 4:} How does Q-VWSD correct the disambiguation results of classical methods?
 \end{itemize}

\subsection{Setup}
All experiments run on a single NVIDIA A30 GPU. Since powerful quantum computing hardware is still limited, we used open-source quantum computing simulation libraries PennyLane \cite{bergholm2018pennylane} for Q-VWSD$_{QC}$. We used the basic-scale \footnote{https://huggingface.co/openai/clip-vit-base-patch32} and large-scale \footnote{https://huggingface.co/openai/clip-vit-large-patch14} versions of the pre-trained model CLIP to encode both images and text. \cite{DBLP:conf/icml/RadfordKHRGASAM21}, respectively. The large version's output reaches 768 dimensions, requiring the simulation of at least a 10-qubit quantum circuit to process this output. This makes it difficult for Q-VWSD$_{QC}$ to run on our current hardware. Therefore, we only use the base version in Q-VWSD$_{QC}$. However, Q-VWSD$_{QI}$ doesn't face this limitation.

Specifically, the pre-trained model we used is based on the publicly available CLIP model from Hugging Face. Q-VWSD$_{QI}$ was run in a Windows environment with Pytorch 1.13, while Q-VWSD$_{QC}$ was executed in a Linux environment, using Pytorch 1.13 and PennyLane 0.31.1. Since the task is unsupervised inference, there was no need to set random seeds or account for the impact of multiple runs. It also does not involve hyperparameter tuning, as all possible parameters are set to their default values. External knowledge was sourced from WordNet 3.1, GPT-3.5-turbo, and GPT-4.0.

% Our code is publicly available \footnote{https://github.com/QuaRobot/Q-VWSD}.

\begin{table*}[h]
\centering
\caption{Main experimental results.}
% \footnotesize
% \resizebox{.8\columnwidth}{!}
\renewcommand\arraystretch{1.2}
\tabcolsep=0.5cm
\begin{threeparttable}

\begin{tabular}{clccccc}
\hline
\multicolumn{2}{c}{\multirow{3}{*}{\textbf{Model}}} & \multirow{3}{*}{\textbf{Source of Glosses}} & \multicolumn{4}{c}{\textbf{SE 23}}                                       \\ \cmidrule(lr){4-7}
\multicolumn{2}{c}{}                                &                                  & \multicolumn{2}{c}{\textbf{Base}} & \multicolumn{2}{c}{\textbf{Large}} \\ \cmidrule(lr){4-5} \cmidrule(lr){6-7}

\multicolumn{2}{c}{}                                &                                  & Hits@1             & MRR              & Hits@1             & MRR             \\ \hline \hline
\multicolumn{2}{c}{CLIP \cite{DBLP:conf/icml/RadfordKHRGASAM21}}                            & -                                & 72.85            & 82.66            & 81.33            & 88.11           \\
\multicolumn{2}{c}{FLAVA \cite{DBLP:conf/cvpr/SinghHGCGRK22}}                           & -                                & 70.13            & 80.67            & -                & -               \\ \hline
\multicolumn{2}{c}{\multirow{5}{*}{CIM \cite{DBLP:conf/acl/KwonGL0023}} }     & WN \cite{fellbaum1998wordnet}                               & 81.98            & 88.83            & 88.06
               &  92.60               \\
\multicolumn{2}{c}{}                                & DG \cite{DBLP:conf/naacl/MalkinLGRJ21}                               & 81.64            & 88.33            & 86.84
                & 91.64
              \\
\multicolumn{2}{c}{}                                & CADG \cite{DBLP:conf/acl/KwonGL0023}                            & 82.65            & 89.28             & 87.83
                & 92.38              \\
\multicolumn{2}{c}{}                                & WN+CADG \cite{DBLP:conf/acl/KwonGL0023}                          & 83.39            & 89.80            & 88.62            & 92.96           \\ 
\multicolumn{2}{c}{}                                & LLM \cite{achiam2023gpt}                         & 77.43            & 85.71            &  86.24            & 91.47           \\ 

\hline

\multicolumn{2}{c}{\multirow{2}{*}{Q-VWSD$_{QC}$ (ours)}} & WN+CADG                          & 84.28 & 90.37            & -            & -           \\
\multicolumn{2}{c}{}                                &LLM                           & \underline{84.61}   & \underline{90.58}   & -   & -   \\\hline

\multicolumn{2}{c}{\multirow{2}{*}{Q-VWSD$_{QI}$ (ours)}} & WN+CADG                          & 84.06            & 90.26            & \underline{88.99}            & \underline{93.26}           \\
\multicolumn{2}{c}{}                                &LLM                           & \textbf{84.99}   & \textbf{90.80}   & \textbf{89.63}   & \textbf{93.59}   \\\hline
\end{tabular}
The best results are in \textbf{bold}, and the second best results are in \underline{underlined}.
\end{threeparttable}
\label{t3}
\end{table*}

\subsection{Datasets}
The dataset \textbf{SE23} we utilized originates from SemEval-2023 VWSD Task \footnote{https://semeval.github.io/SemEval2023/tasks.html}. It consists of 12,896 samples, each consisting of a phrase containing the target word, expressed primarily in English, and 10 candidate images.
Using WordNet \footnote{https://www.nltk.org/howto/wordnet.html}, we further compiled the synsets for each target word and recorded their respective domains. The proportion of each domain within the entire dataset is presented in Fig. \ref{f5} (a). Additionally, We counted the number of glosses for the target words in WordNet, considering this as the number of meanings of each word (i.e., $|G|=s$). The proportion of the number of meanings for the target words in the dataset is shown in Fig. \ref{f5} (b).

\subsection{Baselines}
To our knowledge, there are currently no quantum models available for VWSD. So we chose two representative classes of classical models. 
\begin{itemize}
    \item \textbf{The direct matching method}: \textbf{CLIP} \cite{DBLP:conf/icml/RadfordKHRGASAM21} and \textbf{FLAVA} \cite{DBLP:conf/cvpr/SinghHGCGRK22} are commonly used visual language models that are pre-trained based on contrastive learning. Among them, FLAVA focuses on finer-grained matching information.
    \item \textbf{The classical inference method}: There are few models based on unsupervised VWSD, and the representative work \textbf{CIM} comes from \cite{DBLP:conf/acl/KwonGL0023}. The method is a state-of-the-art model for unsupervised VWSD and can be viewed as a classical version of Q-VWSD.
\end{itemize}
To further explore the impact of gloss on the model, we chose two types of glosses as baselines. These are divided into specialized and non-specialized glosses.
\begin{itemize}    
    \item \textbf{Non-specialized glosses}: \textbf{DG} \cite{DBLP:conf/naacl/MalkinLGRJ21} generates glosses directly from the language model. \textbf{CADG} \cite{DBLP:conf/acl/KwonGL0023} generates glosses from the language model while given the context. \textbf{LLM} generates five glosses from GPT 4.0 \cite{achiam2023gpt} without any given context. 
    \item \textbf{Specialized glosses}: WordNet (\textbf{WN}) \cite{fellbaum1998wordnet} is a language resource tool that combines the knowledge of experts. \textbf{CADG+WN} \cite{DBLP:conf/acl/KwonGL0023} denotes that most of the glosses are sourced with WN, while those words not included in WN are generated by CADG.
\end{itemize}
\subsection{Evaluations}
We follow the evaluation metrics used by the SemEval-2023 VWSD Task:
\begin{itemize}
\item The Mean Reciprocal Ranking (MRR) is given by: 
\begin{equation*}
\textbf{MRR} = \frac{1}{N} \sum_{i=1}^N \frac{1}{\text{rank}_i},
\end{equation*}
where \(N\) is the number of samples and \(\text{rank}_i\) represents the ranking of the posterior probability of the image selected by the model. 

\item The Hit Rate 1 (Hit@1) is calculated as:
\begin{equation*}
\textbf{Hit@1} = \frac{\sum_{i=0}^{N-1} \left( v^*_i == v_{r,i} \right)}{N},
\end{equation*}
where \(v_i\) is the image selected by the model and \(v_{r,i}\) is the correct image.
\end{itemize}

\subsection{Overall Results}
In this subsection, we first aim to address \textbf{RQ 1}: whether Q-VWSD surpasses the baseline methods based on classical probabilities. The results are presented in Table \ref{t3}. Firstly, we observe that matching models like CLIP and FLAVA perform the worst, indicating that the task cannot be simply regarded as a conventional image-text retrieval. Secondly, both Q-VWSD$_{QC}$ and Q-VWSD$_{QI}$ outperform CIM even when using identical glosses. This suggests that introducing a quantum model is beneficial. Finally, Q-VWSD$_{QC}$ and Q-VWSD$_{QI}$ exhibited similar performance, but Q-VWSD$_{QI}$, which is based on heuristic principles, offers greater flexibility, particularly in leveraging the large-scale version of the pre-trained model. As a result, it slightly outperforms Q-VWSD$_{QC}$.

Regarding \textbf{RQ 1}, the effect of semantic bias in glosses can be further illustrated by comparing the performance of CIM and Q-VWSD using LLM and WN+CADG glosses, as shown in Table \ref{t3}. Although large language models are widely regarded as very powerful, the glosses they generate may have severe semantic bias. This leads to CIM using LLM glosses being less effective than WN+CADG glosses. This is because WN+CADG glosses are more specialized and incorporate expert knowledge and a priori information, resulting in fewer semantic biases. In contrast, Q-VWSD somewhat avoids bias through the quantum superposition feature, enabling it to perform better with WN+CADG glosses. Moreover, our method can utilize the rich but biased semantic knowledge from LLM glosses to enhance performance further, which CIM cannot achieve. This demonstrates the advantage of Q-VWSD.

\begin{table*}[th]
    \caption{Accuracy for CIM-WN+CADG, Q-VWSD$_{QI}$-CADG+WN, and Q-VWSD$_{QI}$-LLM.} 
\renewcommand\arraystretch{1.2}
\tabcolsep=0.3cm

\centering
\begin{threeparttable}
\begin{tabular}{cccccccccc}
\hline
Model                     & \multicolumn{3}{c}{CIM-WN+CADG}                      & \multicolumn{3}{c}{Q-VWSD$_{QI}$-CADG+WN}                       & \multicolumn{3}{c}{{Q-VWSD$_{QI}$-LLM}}              \\ 
\cmidrule(lr){2-4} \cmidrule(lr){5-7} \cmidrule(lr){8-10}
$s$                       & Right & Wrong & Accuracy & Right & Wrong & Accuracy &Right & Wrong & Accuracy \\ \hline \hline
2                         & 1316          & 277             & 0.82            & 1325          & 268             & \underline{0.83}            & 1364          & 229             & \textbf{0.85}            \\
3                         & 556           & 173             & 0.76            & 602           & 127             & \underline{0.82}            & 608           & 121             & \textbf{0.83}            \\
4                         & 315           & 84              & 0.78            & 328           & 71              & \underline{0.82}            & 339           & 60              & \textbf{0.84}            \\
5                         & 207           & 63              & 0.76            & 228           & 42              & \underline{0.84}            & 235           & 35              & \textbf{0.87}            \\
6                         & 146           & 59              & 0.71            & 173           & 32              & \underline{0.84}            & 182           & 23              & \textbf{0.88}            \\
7                         & 120           & 34              & 0.77            & 132           & 22              & \underline{0.85}            & 136           & 18              & \textbf{0.88}            \\
8                         & 83            & 23              & 0.78            & 89            & 17              & \textbf{0.83}            & 87            & 19              & \underline{0.82}            \\
9                         & 59            & 12              & \underline{0.83}            & 57            & 14              & 0.80            & 64            & 7               & \textbf{0.90}            \\
$\geq10$                       & 252           & 155             & 0.61            & 344           & 63              & \textbf{0.84}            & 336           & 71              & \underline{0.82}            \\ \hline
\multicolumn{1}{l}{total} & 1568          & 432             & 0.78            & 1669          & 331             & \underline{0.83}            & 1700          & 300             & \textbf{0.85}            \\ \hline
\end{tabular}
 The accuracy is calculated by $Acc.=\frac{\# \ of \ right}{\# \ of \ right+\#  \ of \  wrong}$, where $\# \ of \ right$ and  $\# \ of \ wrong$ denote the number of right disambiguation and the number of wrong disambiguation, respectively. The samples are randomly sampled from polysemy.
\end{threeparttable}
    \label{t5}
\end{table*}

 \begin{figure*}
    \centering
    \includegraphics[width=1\linewidth]{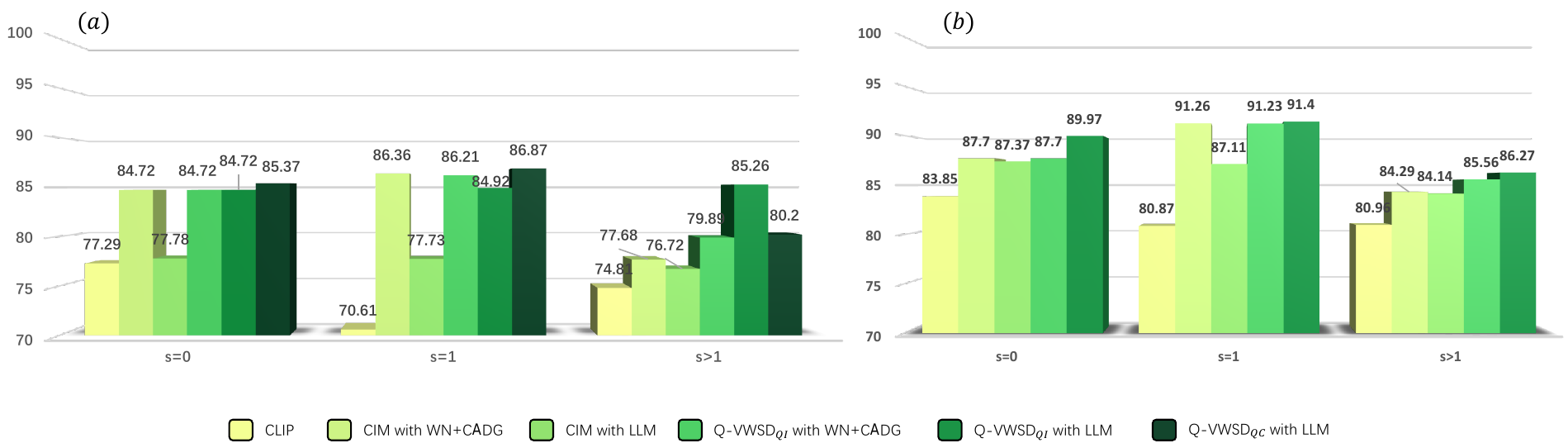}
    \caption{Hits@1 for different methods. (a) shows the performance of the base version of CLIP, while (b) is the larger version.}
    \label{f6}
\end{figure*}

\begin{table}[]
\renewcommand\arraystretch{1.2}
\tabcolsep=0.4cm
\centering
\caption{The results of ablation analysis.}
\begin{threeparttable}
\begin{tabular}{ccccc}
\hline
\multirow{2}{*}{Model} & \multicolumn{2}{c}{Base} & \multicolumn{2}{c}{Large} \\  \cmidrule(lr){2-3} \cmidrule(lr){4-5}
                       & Hits@1        & MRR          & Hits@1        & MRR          \\ \hline \hline
$cos (\theta) = 1$                & 83.97         & 90.22        & 88.93         & 93.22        \\
$cos (\theta) \in [-1,1]$                     & 84.06         & 90.26        & 88.99        & 93.26       \\ \hline \hline
GPT 3.5                     &  84.48
         & 90.53        &  89.46
         &  93.57        \\ 
         GPT 4.0                & \textbf{84.99}   & \textbf{90.80}   & \textbf{89.63}&\textbf{93.59}\\ \hline
\end{tabular}
The best results are in \textbf{bold}, and the second best results are in \underline{underlined}.
\end{threeparttable}
\label{t4}
\end{table}

\begin{table*}[h]
    \caption{Correction rates for CIM-WN+CADG, Q-VWSD$_{QI}$-CADG+WN, and Q-VWSD$_{QI}$-LLM relative to CIM-WN.} 
\renewcommand\arraystretch{1.2}
\tabcolsep=0.02cm

\centering
\begin{threeparttable}
\begin{tabular}{cccccccccc}
\hline
Model                     & \multicolumn{3}{c}{CIM-WN+CADG}                      & \multicolumn{3}{c}{Q-VWSD$_{QI}$-CADG+WN}                       & \multicolumn{3}{c}{{Q-VWSD$_{QI}$-LLM}}           \\ 
\cmidrule(lr){2-4} \cmidrule(lr){5-7} \cmidrule(lr){8-10}
$s$                       & Corrected & Incorrected & Corrected Ratio & Corrected & Incorrected & Corrected Ratio &Corrected & Incorrected & Corrected Ratio \\ \hline \hline
2                         & 105          & 66             & 1.59            & 132          & 51             & \textbf{2.58}            & 119          & 48             & \underline{2.47}            \\
3                         & 49           & 32             & 1.53            & 56           & 23             & \underline{2.43}            & 64           & 16             & \textbf{4.00}            \\
4                         & 32           & 21              & 1.52            & 31           & 13              & \underline{2.38}            & 42           & 16              & \textbf{2.62}            \\
5                         & 17           & 7              & 2.42            & 24           & 4              & \underline{6.00}            & 23           & 3              & \textbf{7.66}            \\
6                         & 12           & 10              & 1.20           & 19           & 8              & \underline{2.37}            & 20           & 7              & \textbf{2.85}            \\
7                         & 11           & 7              & \underline{1.57}            & 11           & 7              & \underline{1.57}            & 17           & 7              & \textbf{2.42}            \\
8                         & 7            & 3              & \underline{2.33}            & 6            & 5              & 1.20            & 3            & 1              & \textbf{3.00}          \\
9                         & 9            & 2              & \underline{4.50}            & 8            & 2              & 4.00            & 10            & 2               & \textbf{5.00}            \\
$\geq10$                       & 37           & 21             & 1.76            & 33           & 15              & \textbf{2.20}            & 44           & 21              & \underline{2.09}            \\ \hline
\multicolumn{1}{l}{total} & 279          & 169             & 1.65            & 320          & 128             & \underline{2.50}            & 342          & 121             & \textbf{2.82}            \\ \hline
\end{tabular}
The correction rate is calculated by $Cor.=\frac{\# \ of \ corrected}{\# \ of \ incorrected}$, where $\# \ of \ corrected$ and $\# \ of \ incorrected$ respectively represent the number of the model corrected CIM-WN wrong disambiguations and the number of the model incorrectly altered CIM-WN right disambiguations.
\end{threeparttable}
    \label{t6}
\end{table*}

\subsection{Ablation Analysis}

To answer \textbf{RQ 2} and further explore the primary contribution of Q-VWSD, we conducted ablation experiments. We first analyzed the disambiguation performance of models for target words with the number of semantic $s=0$, $s=1$, and $s>1$, as shown in Fig. \ref{f6}. When using CADG+WN, for words where $s=0,1$, CIM and Q-VWSD show similar performance because interference items are absent, causing Q-VWSD to degrade to CIM. However, for words where $s>1$, Q-VWSD surpasses CIM. CIM performs only marginally better than CLIP when using large language models, whereas Q-VWSD leverages large language models to achieve its best performance. Particularly for words with $s>1$, large language models significantly enhance the performance of Q-VWSD$_{QI}$. 

Table \ref{t5} provides a more detailed display of the accuracy of Q-VWSD in cases of polysemy, with the 2,000 samples being randomly sampled from polysemous words. It shows the right and wrong disambiguation results for the three models CIM-CADG+WN, Q-VWSD${_{QI}}$-CADG+WN, and Q-VWSD${_{QI}}$-LLM under different semantic numbers. Q-VWSD significantly outperforms CIM in most cases, and it can further leverage large language models to achieve the best results. This indicates that, when faced with polysemy, classical methods are unable to utilize the gloss information of polysemous words effectively, and may instead treat this information as noise, leading to incorrect judgments. In contrast, quantum methods can leverage the gloss information of polysemous words to produce correct results.

\begin{figure}[t]
    \centering
    \includegraphics[width=0.5\linewidth]{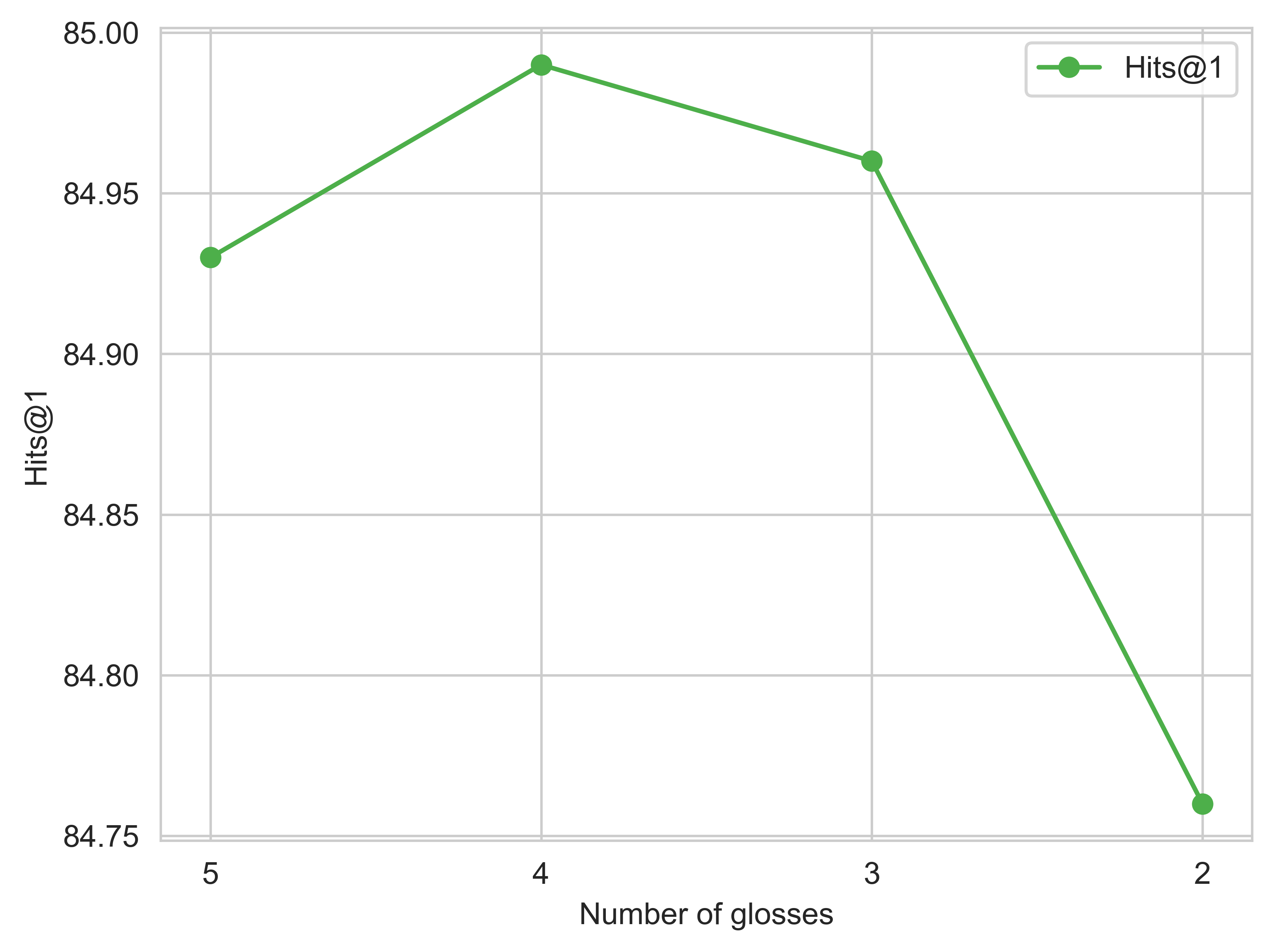}
    \caption{The performance of Q-VWSD with the number of glosses generated by LLMs ranging from 2 to 5.}
    \label{f7}
\end{figure}

To answer \textbf{RQ 3}, a deeper analysis of Q-VWSD$_{QC}$ is required, but it is challenging. Fortunately, its equivalent, Q-VWSD$_{QI}$, explicitly models the interference terms generated during the inference process. This allows us to perform more flexible ablation analysis. Therefore, in the following analysis, Q-VWSD we discuss refers to Q-VWSD$_{QI}$ based on the base version of the pre-trained model. 

We first analyzed the impact of phase on the model within the interference terms, as shown in Table \ref{t4}. Setting $cos(\theta)=1$ in Eq. (\ref{e14}) corresponds to considering only real values, while allowing $cos(\theta)\in[-1, -1]$ introduces complex values. The results indicate that Q-VWSD performs slightly better when considering complex values. We did not observe significant performance improvements, possibly because we assigned cosine similarity values to phases without effectively introducing complex values. Additionally, the table shows the performance of Q-VWSD with glosses provided by GPT-3.5 and GPT-4.0. The results indicate that performance improves with GPT-4.0, which is consistent with expectations. Lastly, we explored the impact of using large language models to generate different numbers of glosses for Q-VWSD, as depicted in Fig. \ref{f7}. The results indicate that the model performs best using up to four glosses, while more glosses decrease performance. This may suggest that most English words have four meanings.

\subsection{Correction Analysis}

\begin{figure}
    \centering
    \includegraphics[width=0.5\linewidth]{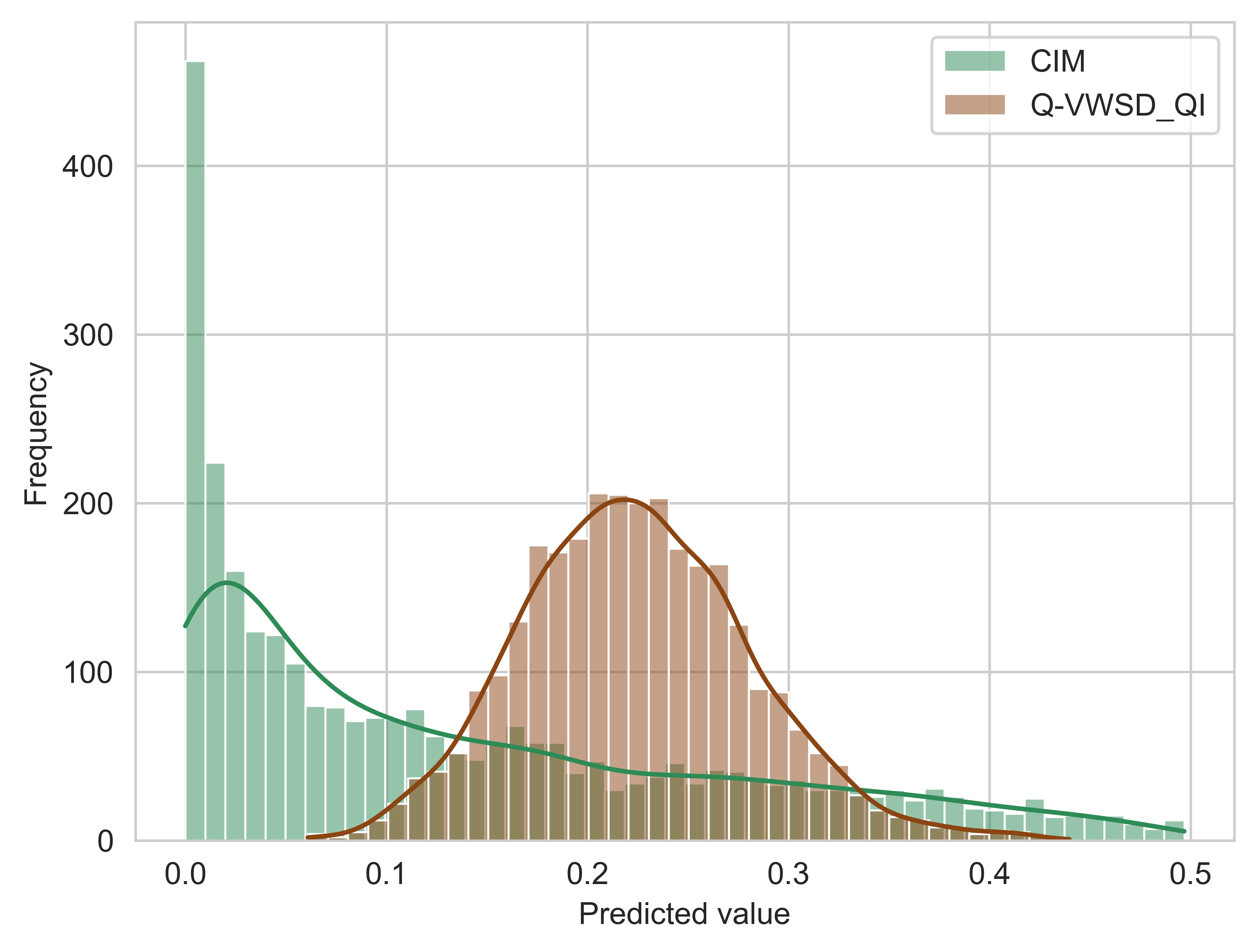}
    \caption{In the samples where CIM made incorrect disambiguations, the output value by both CIM and Q-VWSD.}
    \label{f8}
\end{figure}

In this subsection, to answer \textbf{RQ 4}, we further analyze the correction ability of Q-VWSD compared to CIM, as shown in Fig. \ref{f8}. The output values of the incorrectly disambiguated samples in CIM are primarily concentrated around 0, resembling a long-tail distribution. In contrast, Q-VWSD adjusts these samples to around 0.2, resembling a normal distribution. Therefore, classical probability seems to make ``arbitrary" decisions, whereas quantum probability makes more ``cautious" choices, leading to better performance.

In Table \ref{t6}, we present a more detailed analysis of the correction ability of Q-VWSD. It shows the results of samples that were error disambiguation in CIM-WN and subsequently corrected by models such as CIM-CADG+WN, Q-VWSD$_{QI}$-CADG+WN, and Q-VWSD$_{QI}$-LLM under multi-semantic conditions.
Obviously, Q-VWSD is superior to CIM, especially when combined with large language models. This indicates that Q-VWSD can correct samples where CIM made incorrect disambiguations while ensuring that correctly disambiguation are not changed to incorrect ones. We also note that for semantic $s=7,8$, Q-VWSD does not consistently outperform CIM, likely due to the small number of such samples causing fluctuations. Overall, Q-VWSD still performs better. We believe this is primarily because the interference term models the semantic interactions between glosses, allowing it to correct the posterior probabilities, thereby mitigating semantic bias and leading to more appropriate decisions.

\section{Conclusions}

This paper proposes a quantum inference model for unsupervised VWSD. It mitigates semantic bias by using the superposition state representation of glosses. Based on quantum probability, we reveal that Q-VWSD can generate additional interference terms compared to classical probability methods. This guides the semantic interaction between glosses and is the main reason Q-VWSD can mitigate semantic bias. Experimental results show that Q-VWSD surpasses the latest classical probability methods and effectively utilizes LLM glosses. In contrast, classical probability methods cannot effectively utilize the glosses generated by the LLM, highlighting the importance of addressing semantic bias and the advantages of Q-VWSD.

Quantum machine learning is highly anticipated, but scenarios demonstrating its quantum advantages have yet to materialize. Our work demonstrates the potential of QML in practical scenarios and explores alternative models that can run on classical computing through equivalent heuristic approaches. This alleviates the community's concerns about the practicality of QML in an era where quantum hardware remains immature. At the same time, it also addresses potential issues of lack of rigor that may arise from relying solely on quantum-inspired methods. Our approach is an exploration of the advantages of quantum modeling, demonstrating that quantum machine learning can also inspire us to achieve better artificial intelligence.

\bibliographystyle{unsrt}  
\bibliography{references}  

@inproceedings{DBLP:conf/acl/KwonGL0023,
  author       = {Sunjae Kwon and
                  Rishabh Garodia and
                  Minhwa Lee and
                  Zhichao Yang and
                  Hong Yu},
  editor       = {Anna Rogers and
                  Jordan L. Boyd{-}Graber and
                  Naoaki Okazaki},
  title        = {Vision Meets Definitions: Unsupervised Visual Word Sense Disambiguation
                  Incorporating Gloss Information},
  booktitle    = {Proceedings of the 61st Annual Meeting of the Association for Computational
                  Linguistics (Volume 1: Long Papers), {ACL} 2023, Toronto, Canada,
                  July 9-14, 2023},
  pages        = {1583--1598},
  publisher    = {Association for Computational Linguistics},
  year         = {2023},
  url          = {https://doi.org/10.18653/v1/2023.acl-long.88},
  doi          = {10.18653/V1/2023.ACL-LONG.88},
  bibsource    = {dblp computer science bibliography, https://dblp.org}
}

@inproceedings{DBLP:conf/icml/RadfordKHRGASAM21,
  author       = {Alec Radford and
                  Jong Wook Kim and
                  Chris Hallacy and
                  Aditya Ramesh and
                  Gabriel Goh and
                  Sandhini Agarwal and
                  Girish Sastry and
                  Amanda Askell and
                  Pamela Mishkin and
                  Jack Clark and
                  Gretchen Krueger and
                  Ilya Sutskever},
  editor       = {Marina Meila and
                  Tong Zhang},
  title        = {Learning Transferable Visual Models From Natural Language Supervision},
  booktitle    = {Proceedings of the 38th International Conference on Machine Learning,
                  {ICML} 2021, 18-24 July 2021, Virtual Event},
  series       = {Proceedings of Machine Learning Research},
  volume       = {139},
  pages        = {8748--8763},
  publisher    = {{PMLR}},
  year         = {2021},
  url          = {http://proceedings.mlr.press/v139/radford21a.html},
  bibsource    = {dblp computer science bibliography, https://dblp.org}
}

@inproceedings{DBLP:conf/cvpr/SinghHGCGRK22,
  author       = {Amanpreet Singh and
                  Ronghang Hu and
                  Vedanuj Goswami and
                  Guillaume Couairon and
                  Wojciech Galuba and
                  Marcus Rohrbach and
                  Douwe Kiela},
  title        = {{FLAVA:} {A} Foundational Language And Vision Alignment Model},
  booktitle    = {{IEEE/CVF} Conference on Computer Vision and Pattern Recognition,
                  {CVPR} 2022, New Orleans, LA, USA, June 18-24, 2022},
  pages        = {15617--15629},
  publisher    = {{IEEE}},
  year         = {2022},
  url          = {https://doi.org/10.1109/CVPR52688.2022.01519},
  doi          = {10.1109/CVPR52688.2022.01519},
  bibsource    = {dblp computer science bibliography, https://dblp.org}
}

@inproceedings{DBLP:conf/naacl/MalkinLGRJ21,
  author       = {Nikolay Malkin and
                  Sameera Lanka and
                  Pranav Goel and
                  Sudha Rao and
                  Nebojsa Jojic},
  editor       = {Kristina Toutanova and
                  Anna Rumshisky and
                  Luke Zettlemoyer and
                  Dilek Hakkani{-}T{\"{u}}r and
                  Iz Beltagy and
                  Steven Bethard and
                  Ryan Cotterell and
                  Tanmoy Chakraborty and
                  Yichao Zhou},
  title        = {{GPT} Perdetry Test: Generating new meanings for new words},
  booktitle    = {Proceedings of the 2021 Conference of the North American Chapter of
                  the Association for Computational Linguistics: Human Language Technologies,
                  {NAACL-HLT} 2021, Online, June 6-11, 2021},
  pages        = {5542--5553},
  publisher    = {Association for Computational Linguistics},
  year         = {2021},
  url          = {https://doi.org/10.18653/v1/2021.naacl-main.439},
  doi          = {10.18653/V1/2021.NAACL-MAIN.439},
  bibsource    = {dblp computer science bibliography, https://dblp.org}
}

@inproceedings{raganato2023semeval,
  title={SemEval-2023 task 1: Visual word sense disambiguation},
  author={Raganato, Alessandro and Calixto, Iacer and Ushio, Asahi and Camacho-Collados, Jose and Pilehvar, M and others},
  booktitle={17th International Workshop on Semantic Evaluation, SemEval 2023-Proceedings of the Workshop},
  pages={2227--2234},
  year={2023},
  organization={Association for Computational Linguistics}
}

@inproceedings{DBLP:conf/blackboxnlp/RassinRG22,
  author       = {Royi Rassin and
                  Shauli Ravfogel and
                  Yoav Goldberg},
  editor       = {Jasmijn Bastings and
                  Yonatan Belinkov and
                  Yanai Elazar and
                  Dieuwke Hupkes and
                  Naomi Saphra and
                  Sarah Wiegreffe},
  title        = {{DALLE-2} is Seeing Double: Flaws in Word-to-Concept Mapping in Text2Image
                  Models},
  booktitle    = {Proceedings of the Fifth BlackboxNLP Workshop on Analyzing and Interpreting
                  Neural Networks for NLP, BlackboxNLP@EMNLP 2022, Abu Dhabi, United
                  Arab Emirates (Hybrid), December 8, 2022},
  pages        = {335--345},
  publisher    = {Association for Computational Linguistics},
  year         = {2022},
  url          = {https://doi.org/10.18653/v1/2022.blackboxnlp-1.28},
  doi          = {10.18653/V1/2022.BLACKBOXNLP-1.28},
  bibsource    = {dblp computer science bibliography, https://dblp.org}
}

@inproceedings{zhang2024quantum,
  title={Quantum Interference Model for Semantic Biases of Glosses in Word Sense Disambiguation},
  author={Zhang, Junwei and He, Ruifang and Guo, Fengyu and Liu, Chang},
  booktitle={Proceedings of the AAAI Conference on Artificial Intelligence},
  volume={38},
  number={17},
  pages={19551--19559},
  year={2024}
}

@book{fellbaum1998wordnet,
  title={WordNet: An electronic lexical database},
  author={Fellbaum, Christiane},
  year={1998},
  publisher={MIT press}
}

@inproceedings{huang2019glossbert,
  title={GlossBERT: BERT for Word Sense Disambiguation with Gloss Knowledge},
  author={Huang, Luyao and Sun, Chi and Qiu, Xipeng and Huang, Xuan-Jing},
  booktitle={Proceedings of the 2019 Conference on Empirical Methods in Natural Language Processing and the 9th International Joint Conference on Natural Language Processing (EMNLP-IJCNLP)},
  pages={3509--3514},
  year={2019}
}

@book{von2018mathematical,
  title={Mathematical foundations of quantum mechanics: New edition},
  author={Von Neumann, John},
  volume={53},
  year={2018},
  publisher={Princeton university press}
}

@inproceedings{jiang2020quantum,
  title={A quantum interference inspired neural matching model for ad-hoc retrieval},
  author={Jiang, Yongyu and Zhang, Peng and Gao, Hui and Song, Dawei},
  booktitle={Proceedings of the 43rd International ACM SIGIR Conference on Research and Development in Information Retrieval},
  pages={19--28},
  year={2020}
}

@inproceedings{sordoni2013modeling,
  title={Modeling term dependencies with quantum language models for ir},
  author={Sordoni, Alessandro and Nie, Jian-Yun and Bengio, Yoshua},
  booktitle={Proceedings of the 36th international ACM SIGIR conference on Research and development in information retrieval},
  pages={653--662},
  year={2013}
}

@inproceedings{gkoumas2021quantum,
  title={Quantum cognitively motivated decision fusion for video sentiment analysis},
  author={Gkoumas, Dimitris and Li, Qiuchi and Dehdashti, Shahram and Melucci, Massimo and Yu, Yijun and Song, Dawei},
  booktitle={Proceedings of the AAAI Conference on Artificial Intelligence},
  volume={35},
  number={1},
  pages={827--835},
  year={2021}
}

@book{busemeyer2012quantum,
  title={Quantum models of cognition and decision},
  author={Busemeyer, Jerome R and Bruza, Peter D},
  year={2012},
  publisher={Cambridge University Press}
}

@inproceedings{qiao2024quantum,
  title={A Quantum-Inspired Matching Network with Linguistic Theories for Metaphor Detection},
  author={Qiao, Wenbo and Zhang, Peng and Ma, ZengLai},
  booktitle={Proceedings of the 2024 Joint International Conference on Computational Linguistics, Language Resources and Evaluation (LREC-COLING 2024)},
  pages={1435--1445},
  year={2024}
}

@article{achiam2023gpt,
  title={Gpt-4 technical report},
  author={Achiam, Josh and Adler, Steven and Agarwal, Sandhini and Ahmad, Lama and Akkaya, Ilge and Aleman, Florencia Leoni and Almeida, Diogo and Altenschmidt, Janko and Altman, Sam and Anadkat, Shyamal and others},
  journal={arXiv preprint arXiv:2303.08774},
  year={2023}
}

@article{zhang2022complex,
  title={Complex-valued Neural Network-based Quantum Language Models},
  author={Zhang, Peng and Hui, Wenjie and Wang, Benyou and Zhao, Donghao and Song, Dawei and Lioma, Christina and Simonsen, Jakob Grue},
  journal={ACM Transactions on Information Systems (TOIS)},
  volume={40},
  number={4},
  pages={1--31},
  year={2022},
  publisher={ACM New York, NY}
}

@article{misono1997effects,
  title={Effects and limitations of prosodic and semantic biases on syntactic disambiguation},
  author={Misono, Yasuko and Mazuka, Reiko and Kondo, Tadahisa and Kiritani, Shigeru},
  journal={Journal of Psycholinguistic Research},
  volume={26},
  pages={229--245},
  year={1997},
  publisher={Springer}
}

@inproceedings{kritharoula2023large,
  title={Large Language Models and Multimodal Retrieval for Visual Word Sense Disambiguation},
  author={Kritharoula, Anastasia and Lymperaiou, Maria and Stamou, Giorgos},
  booktitle={Proceedings of the 2023 Conference on Empirical Methods in Natural Language Processing},
  pages={13053--13077},
  year={2023}
}

@inproceedings{ghahroodi2023sut,
  title={SUT at SemEval-2023 task 1: Prompt generation for visual word sense disambiguation},
  author={Ghahroodi, Omid and Dalili, Seyed Arshan and Mesforoush, Sahel and Asgari, Ehsaneddin},
  booktitle={Proceedings of the 17th International Workshop on Semantic Evaluation (SemEval-2023)},
  pages={2160--2163},
  year={2023}
}

@inproceedings{bevilacqua2021recent,
  title={Recent trends in word sense disambiguation: A survey},
  author={Bevilacqua, Michele and Pasini, Tommaso and Raganato, Alessandro and Navigli, Roberto},
  booktitle={International Joint Conference on Artificial Intelligence},
  pages={4330--4338},
  year={2021},
  organization={International Joint Conference on Artificial Intelligence, Inc}
}

@inproceedings{vaiani2023polito,
  title={PoliTo at SemEval-2023 task 1: Clip-based visual-word sense disambiguation based on back-translation},
  author={Vaiani, Lorenzo and Cagliero, Luca and Garza, Paolo},
  booktitle={Proceedings of the 17th International Workshop on Semantic Evaluation (SemEval-2023)},
  pages={1447--1453},
  year={2023}
}

@inproceedings{zhang2023srcb,
  title={SRCB at SemEval-2023 task 1: Prompt based and cross-modal retrieval enhanced visual word sense disambiguation},
  author={Zhang, Xudong and Zhen, Tiange and Zhang, Jing and Wang, Yujin and Liu, Song},
  booktitle={Proceedings of the 17th international workshop on semantic evaluation (SemEval-2023)},
  pages={439--446},
  year={2023}
}

@inproceedings{wei2023stfx,
  title={StFX NLP at SemEval-2023 Task 1: Multimodal Encoding-based Methods for Visual Word Sense Disambiguation},
  author={Wei, Yuchen and King, Milton},
  booktitle={Proceedings of the 17th International Workshop on Semantic Evaluation (SemEval-2023)},
  pages={409--414},
  year={2023}
}

@inproceedings{li2023ecnu_miv,
  title={Ecnu\_miv at semeval-2023 task 1: Ctim-contrastive text-image model for multilingual visual word sense disambiguation},
  author={Li, Zhenghui and Zhang, Qi and Xia, XueYin and Ye, Yinxiang and Huang, Cong},
  booktitle={Proceedings of the 17th International Workshop on Semantic Evaluation (SemEval-2023)},
  pages={101--107},
  year={2023}
}

@inproceedings{yang2023tam,
  title={TAM of SCNU at SemEval-2023 task 1: FCLL: A fine-grained contrastive language-image learning model for cross-language visual word sense disambiguation},
  author={Yang, Qihao and Li, Yong and Wang, Xuelin and Li, Shunhao and Hao, Tianyong},
  booktitle={Proceedings of the 17th International Workshop on Semantic Evaluation (SemEval-2023)},
  pages={506--511},
  year={2023}
}

@inproceedings{dadas2023opi,
  title={OPI at SemEval-2023 Task 1: Image-Text Embeddings and Multimodal Information Retrieval for Visual Word Sense Disambiguation},
  author={Dadas, Slawomir},
  booktitle={Proceedings of the 17th International Workshop on Semantic Evaluation (SemEval-2023)},
  pages={155--162},
  year={2023}
}

@inproceedings{zhang2018end,
  title={End-to-end quantum-like language models with application to question answering},
  author={Zhang, Peng and Niu, Jiabin and Su, Zhan and Wang, Benyou and Ma, Liqun and Song, Dawei},
  booktitle={Proceedings of the AAAI Conference on Artificial Intelligence},
  volume={32},
  number={1},
  year={2018}
}

@article{bergholm2018pennylane,
  title={Pennylane: Automatic differentiation of hybrid quantum-classical computations},
  author={Bergholm, Ville and Izaac, Josh and Schuld, Maria and Gogolin, Christian and Ahmed, Shahnawaz and Ajith, Vishnu and Alam, M Sohaib and Alonso-Linaje, Guillermo and AkashNarayanan, B and Asadi, Ali and others},
  journal={arXiv preprint arXiv:1811.04968},
  year={2018}
}

@article{goodman2015probabilistic,
  title={Probabilistic semantics and pragmatics uncertainty in language and thought},
  author={Goodman, Noah D and Lassiter, Daniel},
  journal={The handbook of contemporary semantic theory},
  pages={655--686},
  year={2015},
  publisher={Wiley Online Library}
}

@article{schuld2021effect,
  title={Effect of data encoding on the expressive power of variational quantum-machine-learning models},
  author={Schuld, Maria and Sweke, Ryan and Meyer, Johannes Jakob},
  journal={Physical Review A},
  volume={103},
  number={3},
  pages={032430},
  year={2021},
  publisher={APS}
}

@article{yu2022power,
  title={Power and limitations of single-qubit native quantum neural networks},
  author={Yu, Zhan and Yao, Hongshun and Li, Mujin and Wang, Xin},
  journal={Advances in Neural Information Processing Systems},
  volume={35},
  pages={27810--27823},
  year={2022}
}

@inproceedings{
zhao2024quantum,
title={Quantum Implicit Neural Representations},
author={Jiaming Zhao and Wenbo Qiao and Peng Zhang and Hui Gao},
booktitle={Forty-first International Conference on Machine Learning},
year={2024},
url={https://openreview.net/forum?id=50vc4HBuKU}
}

@article{rudolph2022generation,
  title={Generation of high-resolution handwritten digits with an ion-trap quantum computer},
  author={Rudolph, Manuel S and Toussaint, Ntwali Bashige and Katabarwa, Amara and Johri, Sonika and Peropadre, Borja and Perdomo-Ortiz, Alejandro},
  journal={Physical Review X},
  volume={12},
  number={3},
  pages={031010},
  year={2022},
  publisher={APS}
}

@article{nakaji2022approximate,
  title={Approximate amplitude encoding in shallow parameterized quantum circuits and its application to financial market indicators},
  author={Nakaji, Kouhei and Uno, Shumpei and Suzuki, Yohichi and Raymond, Rudy and Onodera, Tamiya and Tanaka, Tomoki and Tezuka, Hiroyuki and Mitsuda, Naoki and Yamamoto, Naoki},
  journal={Physical Review Research},
  volume={4},
  number={2},
  pages={023136},
  year={2022},
  publisher={APS}
}

@inproceedings{qiao2024quantum2,
  title={Quantum Topic Model: Topic Modeling Using Variational Quantum Circuits},
  author={Qiao, Wenbo and Zhang, Peng and Zhao, Jiaming and Yang, Chang},
  booktitle={ICASSP 2024-2024 IEEE International Conference on Acoustics, Speech and Signal Processing (ICASSP)},
  pages={5895--5899},
  year={2024},
  organization={IEEE}
}

@inproceedings{DBLP:conf/nips/Bausch20,
  author       = {Johannes Bausch},
  editor       = {Hugo Larochelle and
                  Marc'Aurelio Ranzato and
                  Raia Hadsell and
                  Maria{-}Florina Balcan and
                  Hsuan{-}Tien Lin},
  title        = {Recurrent Quantum Neural Networks},
  booktitle    = {Advances in Neural Information Processing Systems 33: Annual Conference
                  on Neural Information Processing Systems 2020, NeurIPS 2020, December
                  6-12, 2020, virtual},
  year         = {2020},
  url          = {https://proceedings.neurips.cc/paper/2020/hash/0ec96be397dd6d3cf2fecb4a2d627c1c-Abstract.html},
  bibsource    = {dblp computer science bibliography, https://dblp.org}
}

@inproceedings{chen2022quantum,
  title={Quantum long short-term memory},
  author={Chen, Samuel Yen-Chi and Yoo, Shinjae and Fang, Yao-Lung L},
  booktitle={ICASSP 2022-2022 IEEE International Conference on Acoustics, Speech and Signal Processing (ICASSP)},
  pages={8622--8626},
  year={2022},
  organization={IEEE}
}

@inproceedings{DBLP:conf/nips/TangY22,
  author       = {Yehui Tang and
                  Junchi Yan},
  editor       = {Sanmi Koyejo and
                  S. Mohamed and
                  A. Agarwal and
                  Danielle Belgrave and
                  K. Cho and
                  A. Oh},
  title        = {GraphQNTK: Quantum Neural Tangent Kernel for Graph Data},
  booktitle    = {Advances in Neural Information Processing Systems 35: Annual Conference
                  on Neural Information Processing Systems 2022, NeurIPS 2022, New Orleans,
                  LA, USA, November 28 - December 9, 2022},
  year         = {2022},
  url          = {http://papers.nips.cc/paper\_files/paper/2022/hash/285b06e0dd856f20591b0a5beb954151-Abstract-Conference.html},
  bibsource    = {dblp computer science bibliography, https://dblp.org}
}

@inproceedings{DBLP:conf/kdd/YanTY22,
  author       = {Ge Yan and
                  Yehui Tang and
                  Junchi Yan},
  editor       = {Aidong Zhang and
                  Huzefa Rangwala},
  title        = {Towards a Native Quantum Paradigm for Graph Representation Learning:
                  {A} Sampling-based Recurrent Embedding Approach},
  booktitle    = {{KDD} '22: The 28th {ACM} {SIGKDD} Conference on Knowledge Discovery
                  and Data Mining, Washington, DC, USA, August 14 - 18, 2022},
  pages        = {2160--2168},
  publisher    = {{ACM}},
  year         = {2022},
  url          = {https://doi.org/10.1145/3534678.3539327},
  doi          = {10.1145/3534678.3539327},
  bibsource    = {dblp computer science bibliography, https://dblp.org}
}

@article{shin2023exponential,
  title={Exponential data encoding for quantum supervised learning},
  author={Shin, S and Teo, YS and Jeong, H},
  journal={Physical Review A},
  volume={107},
  number={1},
  pages={012422},
  year={2023},
  publisher={APS}
}
\appendix

\end{document}